\documentclass[12pt]{article}

\usepackage{natbib}
\usepackage{times}
\usepackage{amsfonts, amssymb}
\usepackage{setspace}
\usepackage[nohead]{geometry}
\usepackage{amsmath}
\usepackage{amsthm}
\usepackage{graphicx}
\usepackage{authblk}
\newtheorem{definition}{Definition}
\newtheorem{theorem}{Theorem}
\newtheorem{lemma}{Lemma}

\title{Hipsters and the Cool: A Game Theoretic Analysis of Social Identity, Trends and Fads}

\author[1]{Russell Golman\thanks{Corresponding Author; E-mail: rgolman@andrew.cmu.edu}}
\author[2]{Aditi Jain}
\author[2]{Sonica Saraf}
\affil[1]{\small Department of Social and Decision Sciences, Carnegie Mellon University}%, Pittsburgh, PA 15213
\affil[2]{\small Department of Mathematics, Carnegie Mellon University}%, Pittsburgh, PA 15213}

\date{\today}

\begin{document}
\maketitle
%\keywords{Conformity $|$ Games on Social Networks $|$ Popularity Cycles $|$ Social Dynamics $|$ Uniqueness} 
%\onehalfspacing
\doublespacing

%\begin{abstract}
\centerline{{\Large Abstract}} \vspace{1mm}
%{\normalsize
%\begin{sciabstract}
 Cultural trends and popularity cycles can be observed all around us, yet our theories of social influence and identity expression do not explain what perpetuates these complex, often unpredictable social dynamics.  We propose a theory of social identity expression based on the opposing, but not mutually exclusive, motives to conform and to be unique among one’s neighbors in a social network.  We then model the social dynamics that arise from these motives.  We find that the dynamics typically enter random walks or stochastic limit cycles rather than converging to a static equilibrium.  We also prove that without social network structure or, alternatively, without the uniqueness motive, reasonable adaptive dynamics would necessarily converge to equilibrium.  Thus, we show that nuanced psychological assumptions (recognizing preferences for uniqueness along with conformity) and realistic social network structure are both necessary for explaining how complex, unpredictable cultural trends emerge.%}    
%\end{abstract}
%\end{sciabstract}

%\\
\vspace{4mm}
\noindent {\bfseries Keywords}: Conformity $|$ Games on Social Networks $|$ Popularity Cycles $|$ Social Dynamics $|$ Uniqueness

\clearpage
\section*{Introduction}
Popular cultural practices come into and out of fashion.  Researchers have observed boom-and-bust cycles of popularity in music, clothing styles, given names, automobile designs, home furnishings, and even management practices \citep{Shuker,RichardsonKroeber,Reynolds1968,Sproles1981,Berger2008,Berger2009,Lieberson,LiebersonLynn,Robinson1961,Abrahamson1991,Zuckerman2012}. Popularity cycles appear to be driven by social influence, e.g., by people adopting the music that their friends listen to or that they perceive as popular \citep{Salganik06,Salganik08}.  At the individual level, people are constantly looking for new ways to express their preferred social identities \citep{Hetherington,RentfrowGosling,Berger2008,Chan2012}.  The resultant social dynamics do not typically converge to equilibrium.  What are the social forces that lead to such perpetual change and novelty?

Social pressure to conform is a powerful force when behavioral patterns across a society shift in unison.  Psychologists since Asch have recognized the remarkable strength of the conformity motive, stemming from a fundamental goal to fit in as part of a social group \citep{Asch1955,Asch1956,Cialdini1998}.  People tend to feel uncomfortable about considering, holding, and expressing beliefs that conflict with the prevailing views around them as well as about behaving oddly, in ways that might expose oneself as an outsider to the group \citep{AkerlofKranton,PBC}.  Given the conformity motive alone, we might expect to observe convergence to an equilibrium in which society becomes monolithic, yet instead we actually observe persistent diversity.

Opposing the motive to conform is a similarly universal human need for uniqueness \citep{SnyderFromkin,Lynn2002}.  While the desire to differentiate oneself clearly works against the desire to blend in \citep{Imhoff2009}, Chan, Berger and van Boven \citeyearpar{Chan2012} demonstrate that people simultaneously pursue assimilation and differentiation goals, aiming to be identifiable, but not identical \citep[see also][]{Leibenstein,Robinson1961}.  Preferences for idiosyncratic behavioral patterns can preserve diversity \citep{Smaldino2015a}.  Still, the question remains why behavioral patterns often do not remain in a stable equilibrium with everyone finding an optimal balance between distinctiveness and conformity.  Why instead do behavioral patterns go through perpetual change, with particular behaviors cycling into and out of fashion as cultural trends play out? 

Here, we show that along with conformity and uniqueness motives, a realistic network of social interaction is a critical, necessary ingredient for complex social dynamics to emerge.  Specifically, we show that reasonable adaptive dynamics that would necessarily converge to a static equilibrium given random interactions in a well-mixed pool of people instead typically enter random walks or stochastic limit cycles, and thus never converge, when interactions are restricted to individuals' local neighborhoods in their social networks.   

A natural theoretical approach for investigating social influence on decisions is to use game theory. The conformity motive in isolation would create a Keynesian beauty contest, in which what is cool (like what is beautiful) is just what everybody else believes is cool \citep{Keynes}.  The uniqueness motive in isolation would create a congestion game, in which the objective is simply to be distinct from as many other people as possible \citep{Rosenthal73}.  Both games are known to be potential games, for which convergence to a pure strategy Nash equilibrium is practically guaranteed \citep{Monderer1996,Monderer1996a}.  When both motives co-exist and the game is played on a realistic social network, however, the dynamics are more complex. 

Cultural trends can be modeled more realistically as the dynamics of a game on a social network because social influence is mediated by a social network \citep{JacksonZenou}.  Social influence on expressions of individual identity is transmitted whenever an individual observes another person whom he would like to identify with, so the relevant social network is defined by directed connections corresponding to observation.  The connected components of the social network may correspond to distinct social groups, each with its own emergent subculture.  

The desire for uniqueness within one's own social group should not be conflated with a desire for differentiation across groups \citep{Chan2012}.  In models of identity signaling, membership in one group may be preferable to membership in another, and people want to strategically distinguish themselves from those in the less favorable group; e.g., an upper class tries to distinguish itself from the bourgeois while the bourgeois tries to imitate them \citep{Berger2007}.  The dynamic of differentiation and imitation has been hypothesized to lead to fashion cycles \citep{Karni,Pesendorfer95}.  This dynamic does not, however, preserve diversity within groups.  Desire for uniqueness is a necessary part of the explanation.  Our model features in-group conformity and uniqueness motives; it could be augmented with a desire for differentiation across groups, but for parsimony we assume that people care only about their fit within their own groups.

\section*{Model 1: Social Identity Expression in a Well-Mixed Population}
We model the expression of social identity as a game played by a population of $N$ individuals.  Let us say there are $d$ aspects (or dimensions) of identity.  Each person $i$ chooses an expression of his identity $x_i \in \{a..b\}^d$, i.e., represented as a tuple of $d$ integers from some interval.  For example, in the case of choosing a color to wear, three integers between $0$ and $255$ might correspond to shades of red, green, and blue that mix together to form any color.

A person's degree of conformity in the population depends on the distance between his expressed identity and the average (population mean) expression of identity, $\Vert x_i - \bar{x} \Vert$.
%where $\bar{x}=\frac{1}{N}\sum_i x_i$.  
A person's degree of uniqueness in the population depends on the number of others who adopt the exact same expression of identity as him, denoted as $n_i(X)$ where $X$ is the entire population's profile of expressed identities.  Putting together conformity and uniqueness motives, we model person $i$'s utility given the profile of expressed identities as
\begin{equation}
u_i(X) = - \Vert x_i - \bar{x} \Vert^2 - \lambda \, n_i(X)
\label{UtilityEq}
\end{equation} 
where $\lambda$ is a parameter that describes the strength of the uniqueness motive relative to the conformity motive.  This utility function describes a person whose goal is to be similar to everybody, yet the same as nobody.

Over time people may change their expressions of identity to achieve higher utility.  We need not fully prescribe this process, but assume only that people make changes that increase their own utility, in accordance with some {\em better-reply dynamics} \citep{Monderer1996,Friedman}.  
\begin{definition}[Better-reply dynamics]
At any given time $t$, one person $i$ may consider switching from $x_i$ to $x'_i$; he switches if and only if $u_i(X') > u_i(X)$% (where, of course, $x'_j = x_j$ for all $j \neq i$)
; and for each person $i$ and any best response $x^*_i$ (to $X(t)$), the expected time until person $i$ considers switching to $x^*_i$ is finite.     
\end{definition}  
The motivation for better-reply dynamics is that people are boundedly rational and adaptive.  They can see what the people around them are doing and can search for something better (myopically), but they do not instantaneously react to changes in other people's behavior or anticipate these changes before they occur.  Many commonly assumed adaptive learning dynamics are particular specifications of better-reply dynamics. 

\section*{Results: Social Dynamics in a Well-Mixed Population}
\begin{theorem}
Suppose people derive utility from both their conformity and their uniqueness in the population, as in Equation~(\ref{UtilityEq}).  Then any better-reply dynamics necessarily converges to a pure strategy Nash equilibrium.
\label{ConvergenceThm}
\end{theorem}
The proof is presented in the SM Appendix.  It follows from Lemma 1 in the SM Appendix, which identifies an exact potential function for this game.  Two examples of Nash equilibria, among many that exist, are shown in Figure~\ref{fig1:Equilibria}.
%Fig1 Here
\begin{figure*}[!]
	\centering
	%\begin{subfigure}%{\columnwidth}
  %\centering
  \includegraphics[width=8.1cm]{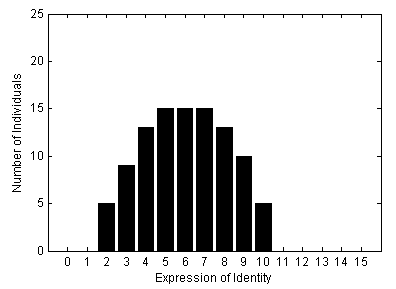}
%  \label{fig1A}
%\end{subfigure}%
%\begin{subfigure}%{\columnwidth}
%  \centering
  \includegraphics[width=8.1cm]{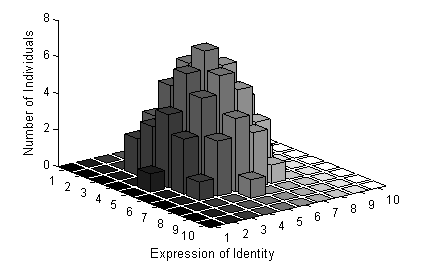}
%  \label{fig1B}
%\end{subfigure}
%		\includegraphics[width=\columnwidth]{Equilibrium1D.png}
	\caption{Two Nash equilibria distributions of identity expression for populations of $N=100$ individuals.  We set $\lambda=1.5$ for this illustration. ({\bfseries A}) One-dimensional identity expression over the domain $\{0..15\}$.  ({\bfseries B}) Two-dimensional identity expression over the domain $\{1..10\}^2$.  By symmetry, the distributions can be shifted anywhere within these (or wider) domains, and many strategy profiles give rise to the same population distributions.  Even after accounting for these symmetries, these Nash equilibria are not unique.}
	\begin{flushleft}
	\vspace{-9.8cm} A \hspace{8.1cm} B \vspace{8.5cm}
	\end{flushleft}
	\label{fig1:Equilibria}
\end{figure*}

Theorem~\ref{ConvergenceThm} says that in a well-mixed population, in the long run we will not see popularity cycles, perpetual change, or novelty.  The fact that we do, in reality, observe popularity cycles, perpetual change, and novelty suggests that we should consider a more realistic model.  We now consider the social dynamics that result from assuming that people care only about the expressed identity of their immediate neighbors in their social network.

\section*{Model 2: Social Identity Expression in Social Networks}
A social network is described by an adjacency matrix $A$ where $a_{ij}=1$ if person $i$ observes, and thus cares about, person $j$'s expressed identity (and equals $0$ if not).  Let $\eta(i) = \{j: a_{ij}=1\}$ denote the set of people that person $i$ observes, i.e., his neighbors. 

Conformity among one's neighbors depends on distance from one's neighbors' average identity, $\bar{x}_{\eta(i)}$.  Uniqueness among one's neighbors depends on the number neighbors who adopt the same expression of identity as oneself, denoted $\tilde{n}_{i}(X;\eta(i))$.  Thus, we now model person $i$'s utility given the profile of expressed identities $X$ and his set of neighbors $\eta(i)$ as
\begin{equation}
u_i(X) = - \Vert x_i - \bar{x}_{\eta(i)} \Vert^2 - \lambda \, \tilde{n}_i(X;\eta(i)).
\label{UtilityOnNetworkEq}
\end{equation}

\section*{Results: Social Dynamics in Social Networks}
\begin{theorem}
Suppose people derive utility from both their conformity and their uniqueness among their neighbors in a social network, as in Equation~(\ref{UtilityOnNetworkEq}) with $\lambda>1$.  Then there exists a social network adjacency matrix $\hat{A}$ such that no pure strategy Nash equilibrium exists and, thus, better-reply dynamics never converge to an absorbing state.
\label{Non-ConvergenceThm}
\end{theorem}
\begin{proof}
By construction.  We provide an example of a social network with $N=3$ people that illustrates the result.  (Any larger social network that contains this network as an out-component also suffices.)  Let person $1$ observe (only) person $2$, person $2$ observe (only) person $3$, and person $3$ observe (only) person $1$.  %(That is, let $a_{12}=a_{23}=a_{31}=1$ be the only nonzero entries in the adjacency matrix.)

Observe that the best response correspondence for each person is as follows:
\begin{align*}
x^*_1 &\in \{x: \Vert x - x_2 \Vert^2 =1 \}     \\
x^*_2 &\in \{x: \Vert x - x_3 \Vert^2 =1 \}			\\
x^*_3 &\in \{x: \Vert x - x_1 \Vert^2 =1 \}.     
\end{align*}
Each person wants to be one unit of distance away from the person he is observing.  However, it is impossible for all three people to simultaneously choose best responses because of the mathematical fact that odd-length cycle graphs are not $2$-colorable. 
\end{proof}

Theorem~\ref{Non-ConvergenceThm} says that with only local interactions in a social network, perpetually changing identity expression and popularity cycles become possible.  Observe that the uniqueness motive is critical for obtaining this result.  If we were to eliminate the uniqueness motive by setting $\lambda=0$, then any homogeneous profile of expressed identities (with $x_i$ identical for all $i$) would be a pure strategy Nash equilibrium, regardless of the social network structure.  The uniqueness motive along with the local interactions together allow for more realistic, complex social dynamics.

Still, Theorem~\ref{Non-ConvergenceThm} only provides an existence result constructed with a highly stylized, simplistic social network.  It does not tell us whether complex social dynamics typically emerge from our model when people are connected by realistic social networks.  We now use computational modeling to explore the dynamics of our model on realistic social networks. 

We used a variant of the Jin-Girvan-Newman algorithm \citep{Jin2001} to create a sample of $100$ directed social networks with a high level of clustering and community structure and limited out-degree ({\itshape Material and Methods}).  For each of these social networks, we repeatedly computed better reply dynamics based on the utility function in Equation~(\ref{UtilityOnNetworkEq}) to see how often the dynamics converged to equilibrium within $30,000$ time steps ({\itshape Material and Methods}).  (We chose the cutoff at $30,000$ time steps based on first computing the dynamics in the full, well-mixed population, for which Theorem~\ref{ConvergenceThm} tells us that they must converge, and finding that across $100$ trials, %the mean time until convergence was $MEAN$ steps and 
the dynamics always converged within $1600$ time steps.)  If the dynamics did not converge within $30,000$ time steps, we classified them as non-convergent (for that trial).  

Figure~\ref{fig2:} shows snapshots of the dynamics on the first social network in our sample between $29,000$ and $30,000$ time steps.  Very quickly (i.e., within just a few hundred time steps) everybody adopts identities in the range $\{0..3\}$, but individuals continually change thereafter.  We can see considerable change in individual expressions of identity in each snapshot. The dynamics do not converge.
%Fig2 Here
\begin{figure*}
	\centering
	%Fits in one row with 3.45cm widths
  \includegraphics[width=5.25cm]{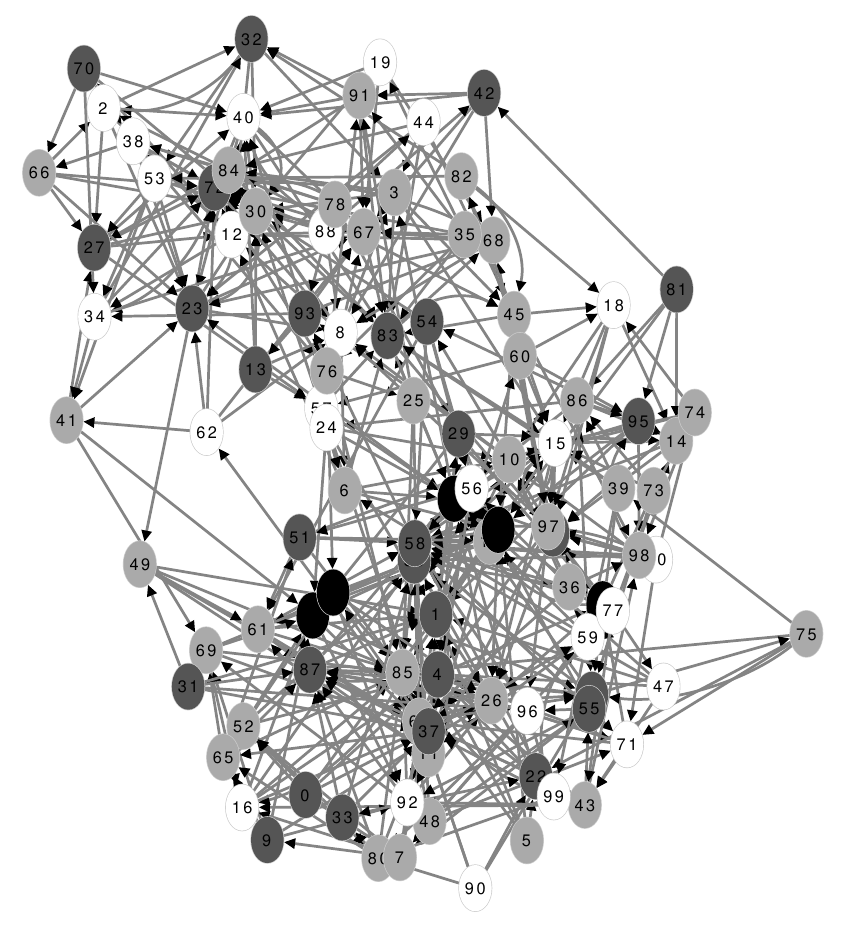}
	%\hspace{.5cm}
	\includegraphics[width=5.25cm]{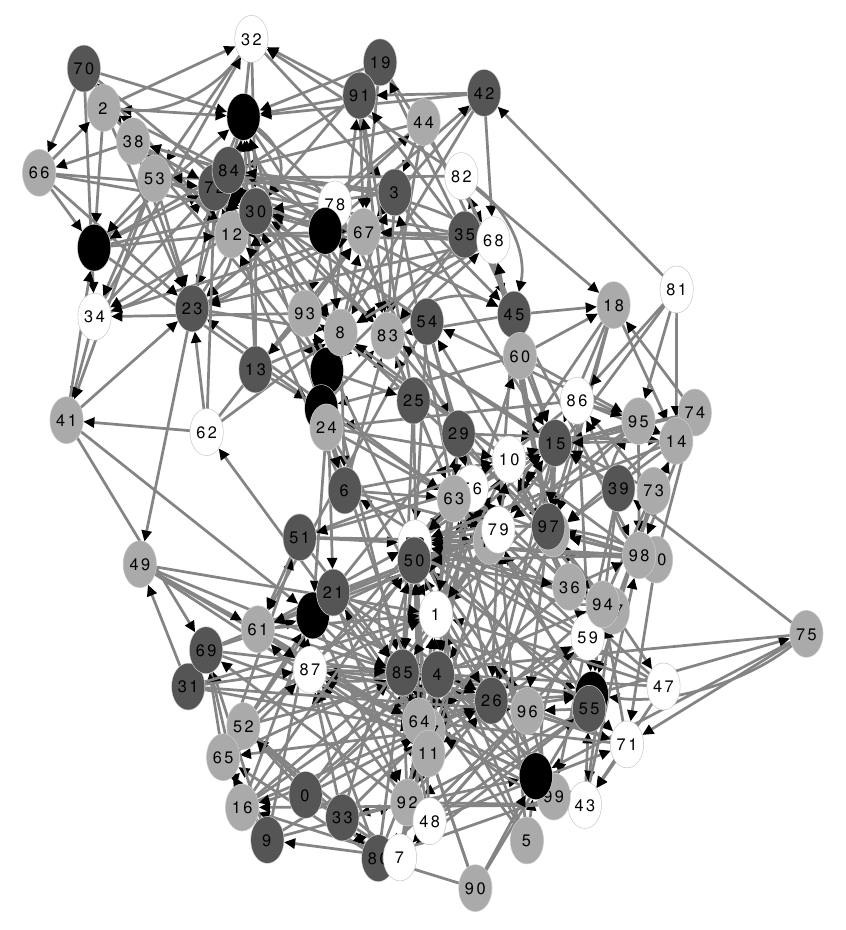}
	%\hspace{.5cm}
	\includegraphics[width=5.25cm]{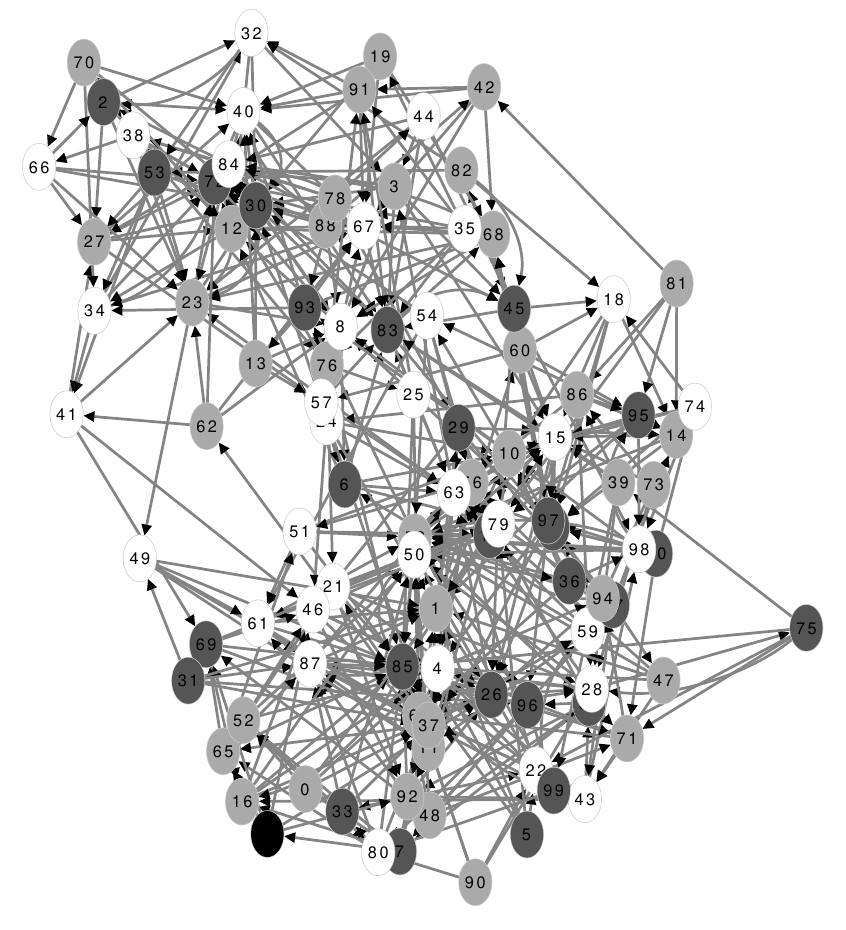}\\
	%\vspace{-.3cm} \hspace{.5cm}
	\includegraphics[width=5.25cm]{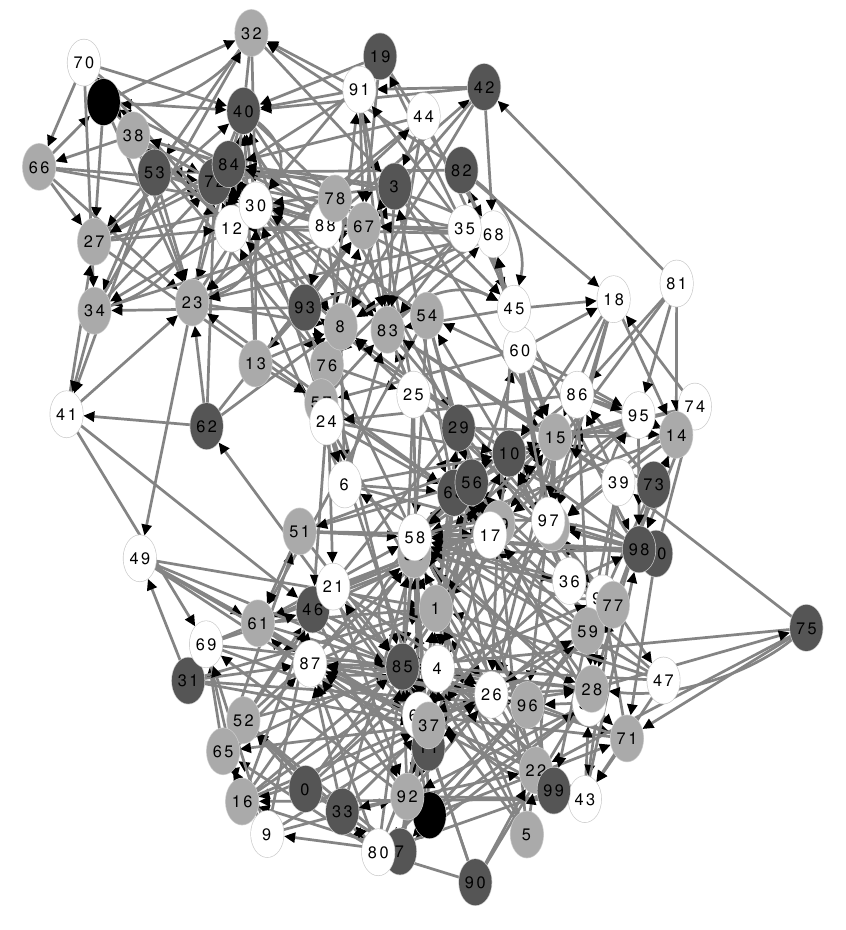}
	\hspace{.5cm}
	\includegraphics[width=5.25cm]{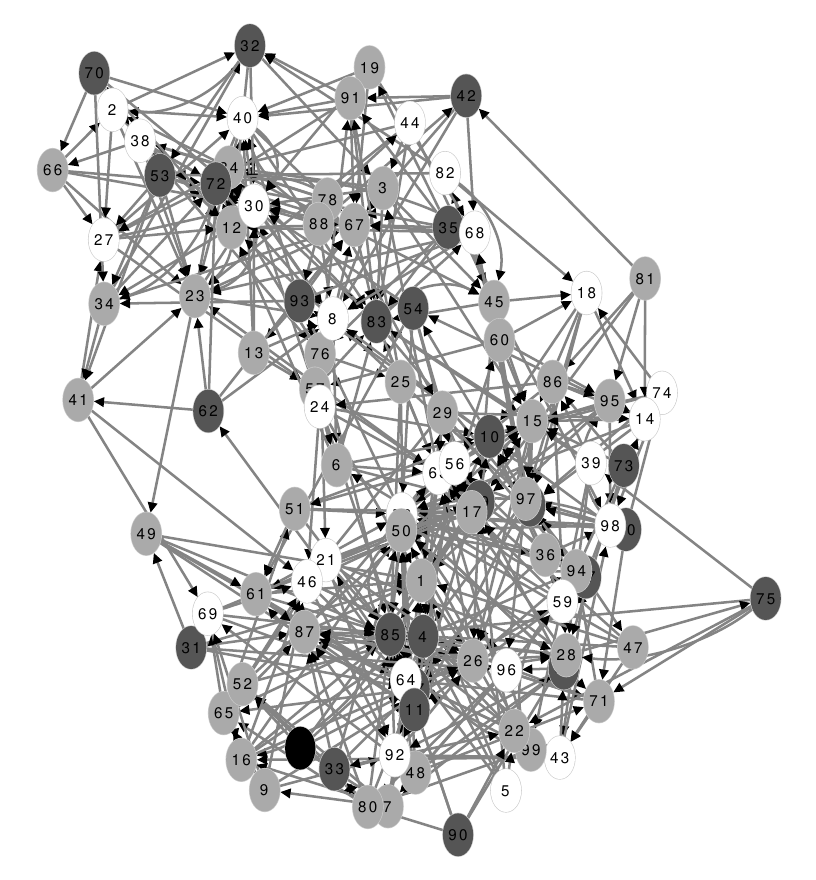}
	\caption{Identity expression in social network $\#1$ changing over time.  Arrows point from an individual to the people he observes.  The shading of the nodes corresponds to each person's expression of identity, from $x_i=0$ if node $i$ is white to $x_i=3$ if node $i$ is black.  Shown from left to right at $29100$, $29300$, $29500$, $29700$ and $29900$ time steps.}
	\label{fig2:}
\end{figure*}
%DISCUSS

On average, across all $100$ social networks in our sample, the dynamics were non-convergent for $99.9\%$ of our trials. Figure~\ref{fig3:Histogram} presents the results of $100$ total trials for each of the $100$ social networks in our sample, showing the number of social networks having particular frequencies of non-convergence.  
%FIGURE3 HERE
\begin{figure}[t]
	\centering
		\includegraphics[width=.5\columnwidth]{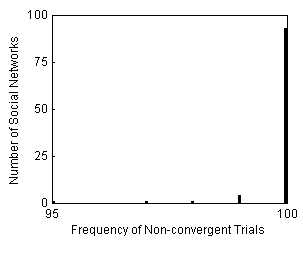}
	\caption{Histogram showing the number of social networks for which we observe particular frequencies of non-convergent trials.}
	\label{fig3:Histogram}
\end{figure}
For each social network in our sample, the dynamics were non-convergent for at least $95$ out of the $100$ trials.  For $93\%$ of the social networks, the dynamics never converged.
%Table~\ref{tab1:} reports the number of non-convergent trials out of $100$ total trials for each of the $100$ social networks in our sample.
%\begin{table}
%	\caption{Frequencies of non-convergent trials}
%		\input{Table_of_Non-Convergent_Trials.tex}
%		%\includegraphics{Table_of_Non-Convergent_Trials.pdf}
%		\label{tab1:}
%\end{table}

These results tell us that with local interactions on realistic social networks, the interplay of conformity and uniqueness motives produces social dynamics for identity expression that are indeed typically non-convergent.  People continually change their expressed identities, and certain forms of expression come into and out of fashion in unpredictable cycles.  Popularity cycles are inherently unpredictable in the model because people typically have multiple better replies (and even multiple best responses) to choose from in the face of most profiles of their neighbors' identity expression.  The multiplicty of paths the dynamics could take leaves room for idiosyncrasy.

The pattern of widespread non-convergence across the entire sample of social networks appears to be robust to variations in the process of search for a better response (i.e., it can be random or sequential), variations in the distribution and average level of out-degree in the social network (short of being fully connected, of course), and variations in the preference parameter $\lambda$ (over the range $\lambda>1$), based on additional trials reported in the SM Appendix.  The social network with ID $37$ is the one that most frequently permits convergence to equilibrium.  Figure~\ref{fig4:} shows this social network and one example of a Nash equilibrium on it.  
\begin{figure}
	\centering
	\includegraphics[width=.5\columnwidth]{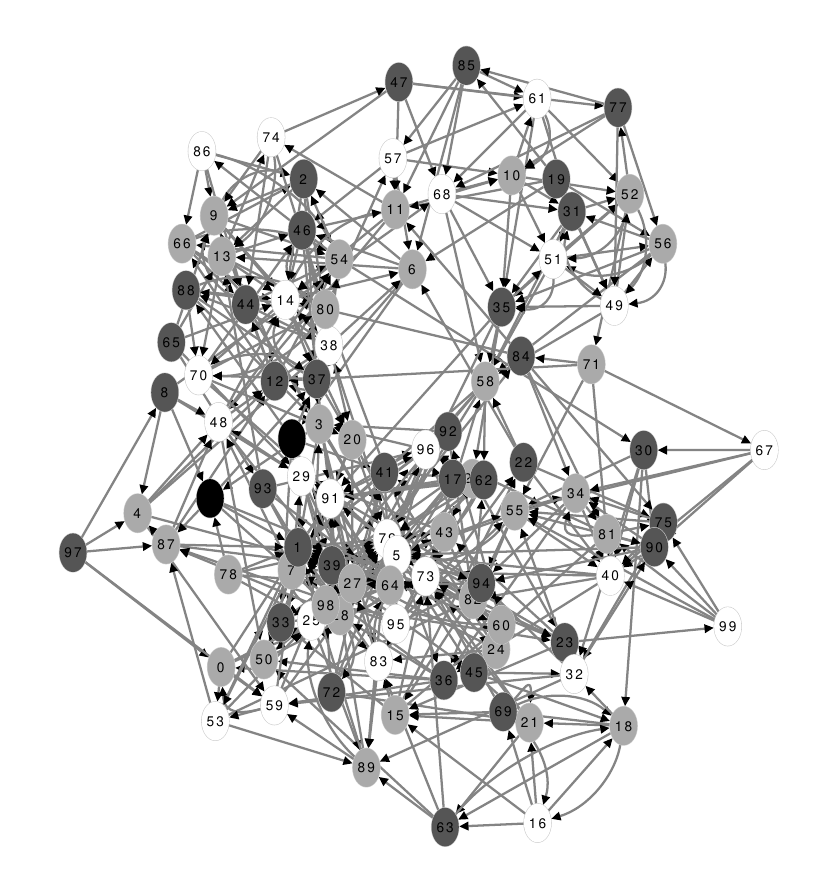}
	\caption{Social network $\#37$ in a Nash equilibrium.  Arrows point from an individual to the people he observes.  The shading of the nodes corresponds to each person's expression of identity, from $x_i=0$ if node $i$ is white to $x_i=3$ if node $i$ is black.}
	\label{fig4:}
\end{figure}

Directed connections in the social network appear to play an important role in obtaining typically non-convergent dynamics.  We explored the better-reply dynamics after inserting reciprocal connections in all of our directed social networks and found that on these (now) undirected social networks, the dynamics converged to equilibrium in $98.7\%$ of our trials.  (The dynamics converged within $3000$ time steps in over $97\%$ of our trials, providing reassurance that findings of non-convergence are fairly robust to allowing the dynamics more time to converge.) Intuitively, directed connections in the social network make it possible that an individual's changing expression of his identity imposes a negative externality on people who observe him, but who he does not notice. The ripple effects may persist or fade, and in more realistic, more complex social networks, they tend to persist indefinitely. 

\section*{Discussion}
Our findings help us understand the role of social networks and local interaction in the dynamics of cultural trends.  Popularity cycles, perpetual change, and novel expressions of social identity should be expected when people observe their neighbors in realistic, directed social networks and care about being unique as well as fitting in.  Such complex social dynamics of identity expression are incompatible with simplistic assumptions disregarding social network structure or reducing social influence to mere conformity pressure absent a desire to individuate oneself.

Recognition of conformity and uniqueness as opposing, but not mutually exclusive, motives is also part of optimal distinctiveness theory \citep{Brewer,Leonardelli2010}.  However, optimal distinctiveness theory posits that people form collective identities by choosing to associate themselves with social groups, whereas our concept of social identity operates at the level of the individual.  In our view, collective identities emerge at the level of the group based on their members' individual identities.  From the alternative, similarly valid perspective, we could propose that individual identities emerge from a psychological process of finding consonance between the collective identities of the many groups that an individual affiliates with at any point in time.  Connecting these perspectives requires deeper understanding of how people choose to associate with or withdraw from social groups, how social network structure endogenously evolves.  While this integration remains beyond our present grasp, we find it useful to have complementary theories aimed at different levels of social identity.     

We use game theory and computational modeling here to describe social dynamics with mathematical precision.  Social phenomena do not always reflect individual preferences \citep{Schelling69,Schelling71}.  Mathematical modeling helps us understand the relationship between individual motives and aggregate social dynamics when interactions generate nontrivial feedbacks.  
Our work here is part of a tradition of formal modeling of social identity and fashion \citep{Bikhchandani1992,Miller1993,StrangMacy,Bettencourt2002,Tassier2004,Acerbi,Simulation2015,Smaldino2015a,Smaldino2015}.  This approach yields us deep theoretical insight, and we hope it inspires more research leading to further insights into   
social dynamics and identity expression.

\section*{Materials and Methods} 
\subsection*{The Social Networks}
We borrow Jin, Girvan, and Newman's Model II algorithm for growing undirected social networks \citep{Jin2001} and modify it to generate directed social networks with $N=100$ people, each of whom can observe up to a maximum of $5$ neighbors.  The network is initialized with all $100$ people and no connections.  The following three steps are then repeated $100$ times:
\begin{enumerate}
\item Choose $3$ pairs of individuals uniformly at random.  For each pair $i$ and $j$, if $i$ observes less than $5$ people and does not already observe $j$, then $i$ begins to observe $j$; else, if $j$ observes less than $5$ people and does not already observe $i$, then $j$ begins to observe $i$.  
\item Randomly select triads $i$, $j$, and $k$ such that $i$ observes $k$ and $k$ observes $j$ or that $i$ and $j$ both observe $k$.  If $i$ observes less than $5$ people and does not already observe $j$, then $i$ begins to observe $j$.  (Real social networks exhibit both patterns of directed closure \citep{BRZ-Romero}.)
\item Randomly select and break $0.5\%$ of connections (rounded up).
\end{enumerate}
All $100$ social networks and the Python source code used to create them will be made available in the SM Appendix.
\subsection*{The Game}
Our computational model adopts the following specification of parameter values for the game: $d=1$; $\{a..b\} = \{0..199\}$; $\lambda=1.5$.
\subsection*{The Better-Reply Dynamics}
Our computational model adopts a specification of the better-reply dynamics in which at each time step, one individual searches for (and upon discovery, adopts) a better reply to the current population profile.  Initial strategies are randomly (uniformly) distributed.  We check for convergence after every $200$ time steps by sequentially checking whether any individual can find a better reply.  In the other time steps, the individual searching for a better reply is randomly selected.  The Python source code and complete output data will be made available in the SM Appendix.

\section*{Acknowledgments}
\noindent \textbf{Funding:} %This work was not supported by outside funding.\\
This research did not receive any specific grant from funding agencies in the public, commercial, or not-for-profit sectors.\\
%\noindent \textbf{Author Contributions:} R. Golman designed research, performed research, and wrote the paper; A. Jain and S. Saraf also performed research, contributing equally.\\
\noindent \textbf{Competing Interests:} The authors declare that they have no competing financial interests.\\
%\noindent \textbf{Data and materials availability:} The source code and complete output data for the computational model as well as the social networks will be made available online as supplementary materials.

%\section*{References}
\bibliographystyle{apalike}%{unsrt}
\bibliography{SocialIdentity}

\clearpage

%\\
%\endnote{{\\
%\bfseries \noindent Author Contributions} \\
%R. Golman designed research, performed research, and wrote the paper; A. Jain and S. Saraf also performed research, contributing equally.}

%\newpage
%\section*{ }
%\newpage
\appendix
\section*{Supplementary Materials}
\subsection*{Formal Definitions}
We can express person $i$'s neighbors' average identity as \[\bar{x}_{\eta(i)} = \frac{1}{|\eta(i)|}\sum_{j \in \eta(i)} x_j.\]
We can express the number of $i$'s neighbors who adopt the same expression of identity as person $i$ as \[\tilde{n}_{i}(X;\eta(i)) = \sum_{j \in \eta(i)} \delta(x_i,x_j),\]
where $\delta$ is the  Kronecker delta function.  In a well-mixed population, we set $\eta(i)=\{j: j \neq i\}$ to recover $n_i(X)$ for all $i$. 
\subsection*{Supplementary Results and Proofs}
\begin{lemma}
In a well-mixed population with utility functions given in Equation~(\ref{UtilityEq}), the game has an exact potential function:
\[\Phi(X) = -\sum_{i=1}^N \frac{N-1}{N} \Vert x_i - \bar{x} \Vert^2 + \frac{1}{2} \lambda \, n_i(X).\]
\label{PotentialFunctionLemma}
\end{lemma}
\begin{proof}
Consider a change in the profile of identities $X \rightarrow X'$ resulting from person $i$ alone changing his identity $x_i \rightarrow x'_i$, i.e., such that $x'_j = x_j$ for all $j \neq i$.  We need only show that the change in the potential function equals the change in $i$'s utility: $\Phi(X') - \Phi(X) = u_i(X') - u_i(X)$.

We express the change in the potential function as a sum of the changes in each term:
\begin{multline*}\Phi(X') - \Phi(X) = \\ \sum_{j=1}^N \frac{N-1}{N} \left(\Vert x_j - \bar{x} \Vert^2 - \Vert x'_j - \bar{x}' \Vert^2\right) +  \sum_{j=1}^N\frac{1}{2} \lambda \, \left(n_j(X) - n_j(X')\right).
\end{multline*}
We consider each of the two summations separately.  

We expand the first sum:
\begin{multline}
\sum_{j=1}^N \frac{N-1}{N} \left(\Vert x_j - \bar{x} \Vert^2 - \Vert x'_j - \bar{x}' \Vert^2\right) = \\
\frac{N-1}{N} \left(\Vert x_i - \bar{x} \Vert^2 - \Vert x'_i - \bar{x}' \Vert^2\right) + \\ \sum_{j \neq i} \frac{N-1}{N} \left(\Vert x_j - \bar{x} \Vert^2 - \Vert x'_j - \bar{x}' \Vert^2\right).
\label{expansionoffirstsum}
\end{multline}
We find it useful to express the average identity as $\bar{x} = \frac{N-1}{N}\bar{x}_{-i} + \frac{1}{N}x_i$.
Plugging in to the first term in Equation~(\ref{expansionoffirstsum}), we have:
\[\Vert x_i - \bar{x} \Vert^2 - \Vert x'_i - \bar{x}' \Vert^2 = \left(\frac{N-1}{N}\right)^2 \left(\Vert x_i - \bar{x}_{-i} \Vert^2 - \Vert x'_i - \bar{x}_{-i} \Vert^2\right).\]
Plugging in to the second term in Equation~(\ref{expansionoffirstsum}), expanding and canceling off common terms, we have for any $j \neq i$:
\begin{multline*} \Vert x_j - \bar{x} \Vert^2 - \Vert x'_j - \bar{x}' \Vert^2 = \\ \frac{1}{N^2} \left(\Vert x_i - \bar{x}_{-i} \Vert^2 - \Vert x'_i - \bar{x}_{-i} \Vert^2\right) + \frac{2}{N} (x_j - \bar{x}_{-i}) \cdot (x_i - x'_i). \end{multline*}
Observe that the last term here drops out when we sum over all $j \neq i$ because $\sum_{j \neq i} (x_j - \bar{x}_{-i}) = 0$.  The first term does not depend on $j$, so summing over all $j \neq i$ just multiplies this term by a factor of $(N-1)$.  Putting it all together, we find that Equation~(\ref{expansionoffirstsum}) simplifies to:
\begin{multline}
\sum_{j=1}^N \frac{N-1}{N} \left(\Vert x_j - \bar{x} \Vert^2 - \Vert x'_j - \bar{x}' \Vert^2\right) \\
= \left(\frac{(N-1)^3}{N^3} + \frac{(N-1)^2}{N^3}\right) \left(\Vert x_i - \bar{x}_{-i} \Vert^2 - \Vert x'_i - \bar{x}_{-i} \Vert^2\right) \\
=\left(\frac{N-1}{N}\right)^2 \left(\Vert x_i - \bar{x}_{-i} \Vert^2 - \Vert x'_i - \bar{x}_{-i} \Vert^2\right) \\ 
=\Vert x_i - \bar{x} \Vert^2 - \Vert x'_i - \bar{x}' \Vert^2.
\label{ConformityPartEq}
\end{multline}  

Now, returning to the second part of the change in the potential function, we can use the formal definition of $n_j(X)$ to write:
\[\sum_{j=1}^N\frac{1}{2} \lambda \, \left(n_j(X) - n_j(X')\right) = \frac{1}{2} \lambda \sum_{j=1}^N \sum_{k \neq j} \left(\delta(x_j,x_k) - \delta(x'_j,x'_k)\right).\]
The terms cancel whenever $j \neq i$ and $k \neq i$, so we are left with:
\begin{multline}
\sum_{j=1}^N\frac{1}{2} \lambda \, \left(n_j(X) - n_j(X')\right) = \\ 
\frac{1}{2} \lambda \left(\sum_{j\neq i} \left(\delta(x_j,x_i) - \delta(x'_j,x'_i)\right) + \sum_{k\neq i} \left(\delta(x_i,x_k) - \delta(x'_i,x'_k)\right)\right) \\ 
= \lambda \sum_{j\neq i} \left(\delta(x_j,x_i) - \delta(x'_j,x'_i)\right) \\
= \lambda \left(n_i(X) -n_i(X')\right).
\label{UniquenessPartEq}
\end{multline}

Putting Equations~(\ref{ConformityPartEq}) and~(\ref{UniquenessPartEq}) together, we have now shown that $\Phi(X') - \Phi(X) = u_i(X') - u_i(X)$.
\end{proof}
\subsubsection*{Proof of Theorem~\ref{ConvergenceThm}}
Theorem~\ref{ConvergenceThm} now follows from Lemma~\ref{PotentialFunctionLemma} by Monderer and Shapley's argument \citeyearpar{Monderer1996}.\qed

\end{document}